\documentclass[copyright,creativecommons]{eptcs}
\providecommand{\event}{AFL 2026} 

\usepackage{amssymb}
\usepackage{amsmath}
\usepackage[capitalise,nameinlink]{cleveref}
\usepackage{tikz}
\usetikzlibrary{shapes.geometric,backgrounds,calc}

\newcommand{\IN}{\mathbb N}

\newcommand\ma[1]{{\cal#1}}

\newcommand\pda{\textsc{pda}}
\newcommand\pdas{\textsc{pda}s}

\newcommand\oca{\textsc{oca}}
\newcommand\ocas{\textsc{oca}s}

\newcommand{\strong}{{\sf strong}}
\newcommand{\accept}{{\sf accept}}
\newcommand{\weak}{{\sf weak}}

\newcommand{\Qincr}[1]{\mbox{$#1^{\scriptstyle\uparrow}$}}
\newcommand{\Qdecr}[1]{\mbox{$#1_{\scriptstyle\downarrow}$}}
\newcommand{\Qstay}[1]{\mbox{$\overline#1$}}

\newcommand*{\qed}{\mbox{}\nolinebreak\hfill~\raisebox{0.77ex}[0ex]{\framebox[1ex][l]{}}}

\newtheorem{theorem}{Theorem}

\newtheorem{lemma}{Lemma}
\newtheorem{corollary}{Corollary}

\newtheorem{definition}{Definition}
\newenvironment{proof}{\noindent{\em Proof.}}{\bigskip\noindent}

\newtheorem{ex}{Example}

\newenvironment{example-cont}[1]{\bigskip\noindent\textbf{Example~\ref{#1}.~(cont.)\hspace{\labelsep}}}{\bigskip\noindent}

\title{Turn Complexity and Bounded Languages}
\author{Giovanni Pighizzini
\institute{Dipartimento di Informatica\\
Universit\`{a} degli Studi di Milano, Italy}
\email{pighizzini@di.unimi.it}
}
\def\titlerunning{Turn Complexity and Bounded Languages}
\def\authorrunning{G.~Pighizzini}

\begin{document}
\maketitle

\begin{abstract}
	A \emph{turn} or \emph{reversal} in a computation of a pushdown automaton is a
	switch from a phase in which the height of the pushdown store increases
	to a phase in which it decreases.
	Given a pushdown automaton, we first consider, for each string in its language, the minimum number of turns made in accepting computations
	(\weak\ measure).
	We prove that it is decidable whether a pushdown automaton accepts a bounded language in a constant number of turns and whether it accepts a bounded language  in~$k$ turns,
	for any given~$k\geq 0$. This is in contrast to the general case, in which these problems
	are known to be undecidable, with the exception of acceptance in~$0$ turns, which is decidable.
	Furthermore, we prove that when the number of turns sufficient to accept a bounded language is not limited by any constants,
	it  linearly grows with respect to the input length. Also this is in contrast with the general
	case where, for each nonnegative~$k$, there exists a language for which the number of turns
	necessary and sufficient is of the order of~$\log^{(k)}$, the~$k$ times composition of the logarithm
	with itself.
	We also prove that, when the costs of all accepting computations are taken into account (\accept\ measure), a linear lower bound for the
	number of turns, if not limited by any constants, holds even removing the restriction to bounded languages.
\end{abstract}

\section{Introduction}

In previous works, we studied different complexity measures for pushdown automata.
The \emph{pushdown height} takes into consideration the maximal height reached by the pushdown store in a computation~\cite{PP23,PP25},
the \emph{push complexity} counts the number of push operations~\cite{Pig26}, while the \emph{turn complexity} considers the number of switches
of the pushdown from an increasing to a decreasing phase~\cite{Pig25}.\footnote{%
	As in the paper by Ginsburg and Spainer~\cite{GS66}, in which this notion was investigated for the first time, and in our
	previous work on this topic~\cite{Pig25}, here we use the term~\emph{turn}.
	Sometimes turns are called~\emph{reversals} (e.g.,~\cite{FKMMZ23}).
}

For pushdown height and push complexity, a constant bound immediately implies the possibility of encoding the pushdown store in a finite
control and, hence, the regularity of the accepted language. 
In case these measures are not bounded by any constants, they grow at least as a double logarithmic function in the length of the input
if, for each accepted string, only the cost of the \emph{less expensive accepting computation} (\weak\ measure) is taken into account.
Furthermore, there is a machine which uses exactly such an amount. Hence, these lower bounds cannot be improved,
except in the case of \emph{unary languages}, i.e., languages defined over a one-letter input alphabet, in which the optimal lower bounds 
in the noncostant case are logarithmic for pushdown height and linear for push complexity~\cite{PP23,Pig26}.
We point out that the idea of considering, for each accepted input, the less expensive accepting computation is related to the idea of
nondeterminism in a broad sense: when more choices are possible, a nondeterministic machine is able not only to select the one that leads to acceptance,
if any, but among all possibilities leading to acceptance the machine is able to select the less expensive one, with respect to resource under consideration.

In contrast, when considering the costs of \emph{all accepting computations}, i.e., for each accepted input the cost of most expensive accepting computation (\accept\ measure),
the previous lower bounds for pushdown height and push complexity become linear, regardless the size of the input alphabet.
As a consequence, also in the deterministic case we obtain linear lower bounds~\cite{PP25}.

In the case of turn complexity, the situation is very different.
First of all, pushdown automata using a constant number of turns accept the class of \emph{ultralinear languages}~\cite{GS66},
a proper subclass of context-free languages, larger than the class of regular languages. Indeed, already one turn gives to pushdown automata more power than finite automata.
Concerning the nonconstant case, considering the \weak\ measure, we proved the existence of an infinite hierarchy:
for each integer~$k\geq 0$ there exists a language that can be accepted using a number of turns of the order of~$\log^{(k)} n$,
but cannot be accepted using a number of turns of the order of~$\log^{(k+1)} n$ under the \weak\ measure (here, $\log^{(k)}$ denotes the~$k$-times
composition of the logarithmic function with itself). Furthermore, there exists a nonregular language which is accepted in~$\log^*n$ turns
(the iterated logarithm)~\cite{Pig25}.

In the witness pushdown automata used to prove these results, a fundamental role is played by the nondeterminism combined with the
structure of the strings in the accepted languages. So, it is natural to ask what happens if we limit the role of the nondeterminism, by
considering the costs of all accepting computations (\accept\ measure)\footnote{%
	According to the previous observation of the interpretation of the nondeterminism in the \weak\ measure, 
	in the \accept\ measure the role of the nondeterminism is restricted only to the selection of the accepting strategy, but it is not used to select
	the less expensive one.%
} or the costs for deterministic machines only,
or if we remove or limit the possibility to have special structures encoded in the accepted strings.
Concerning the latter point, we could ask what happens in the case of \emph{pushdown automata} with a unary input alphabet.
Here, we attack and solve this problem in a more general case. We consider machines accepting \emph{word-bounded languages}, i.e.,
subsets of~$w_1^*w_2^*\cdots w_m^*$, where~$w_1,w_2,\ldots,w_m$ are fixed words over a given alphabet.
We prove that for pushdown automata accepting word-bounded languages, under the \weak\ measure, if the number of turns is nonconstant
then it linearly grows with respect to the input length.
Concerning the former point, we also prove a linear lower bound for the number of turns under the \accept\ measure, when nonconstant,
even if the restriction to word-bounded languages is dropped.

\medskip

%
Before considering these complexity questions, we investigate the problem of deciding whether a pushdown automaton accepts in a finite number
of turns and when it accepts in~$k$ turns, for any fixed~$k$.
Under the \accept\ measure these problems are decidable~\cite{GS66}.
In contrast, under the \weak\ measure they are undecidable, even if the machine is a one-counter automaton, with the exception of acceptance
in~$0$ turns~\cite{Pig25}.
The arguments used to prove these undecidability results derive from classical arguments related to encodings of Turing 
machine computations and do not work in the case the accepted language is restricted to be unary or, more in general, to be bounded.
So it is  natural to ask what happens under these restrictions.
In the paper, we investigate this question and we prove that the above-mentioned problems are decidable, 
namely it is decidable whether a pushdown automaton accepts a word-bounded language in a finite number of turns and when it accepts a word-bounded language in~$k$ turns, for any fixed~$k$.

In the general case, similar undecidability results for pushdown and one-counter automata have been proved, under the \weak\ measure, considering pushdown height and 
push complexity, while decidability has been obtained for the restriction to unary languages~\cite{PP23,PP25,Pig26}.
Results concerning the number of turns in pushdown automata and in the restricted case of pushdown automata accepting bounded language 
are presented in~\cite{Mal07,MP13}, also including relationships between the number of turns and the size of the descriptions of machines.

%

\medskip

The tools we use to prove our decidability and complexity results derive from the theory of \emph{semilinear sets}  
and \emph{Presburger Arithmetic}, and the well-known fact that Parikh images of context-free languages are
semilinear~\cite{GS66b,Parikh1966}. 
In particular, for the results concerning the \weak\ measure and the word-bounded case, 
we use the fact that there is a one-to-one correspondence between strings in a strictly
letter-bounded language accepted by a pushdown automaton and the vectors in its Parikh image.\footnote{%
	A \emph{strictly letter-bounded language} is a subset of~$a_1^*a_2^*\cdots a_m^*$, where~$a_1,a_2,\ldots,a_m$
	are pairwise different symbols.
}
We first solve the strictly letter-bounded case and then we reduce the word bounded case to it.
For the results concerning the \accept\ case, we are interested only in the length of strings, so we do
not need such correspondence. Hence, by adapting the same techniques, we are able to formulate them
by removing the restriction to word bounded-languages.

\section{Preliminaries}
\label{sec:prel}

We assume that the reader is familiar with the standard notions of automata and formal language theory as presented in textbooks, e.g.,~\cite{HU79}.
Given an alphabet~$\Sigma$, the set of strings over~$\Sigma$ is denoted by~$\Sigma^*$, with the empty
string denoted by~$\varepsilon$. The length of a string~$x\in\Sigma^*$ is denoted as~$|x|$, while the number of occurrences of a symbol~$a\in\Sigma$ in~$x$ as~$|x|_a$. 
Given a set~$S$, its cardinality is denoted as~$\# S$, the family of
its subsets as~$2^S$, its complement as~$S^c$.
Given a string~$w\in\Sigma^*$ and a subset~$\Sigma'$ of~$\Sigma$, let us denote by~$\pi_{\Sigma'}(x)$ the string obtained by deleting
from~$x$ all the symbols that do not belong to~$\Sigma'$.

A language~$L\subseteq\Sigma^*$ is said to be a \emph{word-bounded language} (or simply \emph{bounded language}) if it is a subset 
of~$w_1^*w_2^*\cdots w_m^*$, for some~$w_1,w_2,\ldots,w_m\in\Sigma^*$.
If each of these strings consists only of one symbol, i.e.,~$w_1,w_2,\ldots,w_m\in\Sigma$, then~$L$ is said to be a \emph{letter-bounded language}.
Furthermore, if these symbols are pairwise different, then~$L$ is a \emph{strictly letter-bounded} language.

\subsection{Pushdown automata and turns}

We shortly remember the notion of pushdown automata as presented in~\cite{GS66,HU79}.
A \emph{pushdown automaton} (\pda, for short) is defined as~$\ma{M}=(Q,\Sigma,\Gamma,\delta,q_0,Z_0)$,
where~$Q$, $\Sigma$, and $\Gamma$ are finite sets: the set of states, the input alphabet, the pushdown alphabet, respectively;
$q_0\in Q$ is the initial state; $Z_0\in\Gamma$ is the start symbol on the pushdown store;
$\delta$ is the \emph{transition function} from~$Q\times(\Sigma\cup\{\varepsilon\})\times\Gamma$ to the finite subsets of~$Q\times\Gamma^*$,
where~$(p,\alpha)\in\delta(q,\sigma,A)$, with~$p,q\in Q$, $\sigma\in\Sigma\cup\{\varepsilon\}$, $A\in\Gamma$, $\alpha\in\Gamma^*$,
means that the automaton in the state~$q$, with~$A$ on the top of the pushdown store, reading the symbol~$\sigma\in\Sigma$ from
the tape, or without reading any symbol when~$\sigma=\varepsilon$, can move in the state~$p$ after replacing~$A$ on the top of the
stack with the string~$\alpha$.\footnote{%
	We use the convention that pushdown strings are written from top to bottom; hence, when~$\alpha\neq\varepsilon$, the leftmost symbol of~$\alpha$ will be
	on the top of the stack after this transition.
}
We assume \emph{acceptance by empty store}, namely an input~$w$ is accepted when there exists a computation starting from
the initial configuration  (input head on the first symbol of~$w$, finite control in the state~$q_0$, pushdown store containing only~$Z_0$)
that, according to the transition function, reaches a configuration in which all the input has been scanned (i.e., the input head is to 
the right of the rightmost symbol of~$w$) and the pushdown store is empty.
It is well known that the class of language accepted by these devices coincides with the class of 
\emph{context-free languages}. 


A \emph{turn} in a computation~$\cal C$ of a \pda\ is a sequence of~$k\geq 2$ moves
such that the height of the pushdown increases in the first move, decreases in the last move,
and it does not change in between, namely a turn is defined by~$k+1$ consecutive configurations, in which the
pushdown contents are~$\gamma_0, \gamma_1, \ldots, \gamma_k$, and~$|\gamma_0|<|\gamma_1|=|\gamma_2|=\cdots=|\gamma_{k-1}|>|\gamma_k|$.

According to~\cite{GS66}, a \pda~$\ma{M}$ is said to be~\emph{$k$-turn}, for an integer~$k\geq 0$, when \emph{each accepting computation} of~$\ma{M}$ contains at
most~$k$ turns.  A \pda~$\ma{M}$ is \emph{finite turn} if it is~$k$-turn for some~$k\geq 0$.
It is easy to see that~$0$-turns \pdas\ recognize only regular languages.
Furthermore, the class of languages accepted by~$1$-turn \pdas\ is the class of~\emph{linear languages}. 
This class is strictly included in the class of languages accepted by finite-turn \pdas, which is the class of 
\emph{ultralinear languages}, a proper subclass of context-free languages.

In the above definitions, the number of turns in \emph{each accepting computation} is taken into account. This measure is called \accept\ measure.  If we consider the number of turns in
\emph{all computations}, regardless the fact that they are accepting or not, we obtain
the \strong\ measure.

Since in a nondeterministic machine several computations using different numbers of turns could be possible, we will mainly focus on a different notion, which is called \weak\ measure. 
In this measure, we are going to consider the number of turns
which are \emph{sufficient} to accept input strings.
More precisely, we give the following definition for the \weak\ measure:

\begin{definition}
\label{def:turns}
	Let~$\ma{M}$ be a \pda\ accepting a language~$L$.
	We say that \emph{a string~$w\in L$ is accepted in~$k$ turns}, for an integer~$k\geq 0$, 
	when there exists a computation of~$\ma{M}$ accepting~$w$ and containing at most~$k$ turns.
	
	Given a function~$t:\IN\rightarrow\IN$, we say that \emph{$\ma{M}$ accepts~$L$ in~$t(n)$ turns} when, for each~$w\in L$
	with~$|w|=n$, there is an accepting computation on~$w$ containing at most~$t(n)$ turns. 
	Furthermore, if the function~$t(n)$ is bounded by a constant, then we say that \emph{$\ma{M}$ accepts in a finite number of turns}.
\end{definition}

If we know that a \pda\ accepts in~$k$ turns (\weak\ measure), for some given~$k\geq 0$, then we can obtain from it an equivalent~$k$-turn \pda\ (\accept\ measure) just
counting in the finite control the number of turns and stopping the computation when such a number exceeds~$k$. Actually, the resulting automaton makes exactly~$k$ turns in each computation, hence it satisfies also the \strong\ measure.
However, acceptance in~$k$ turns cannot be decided, unless~$k=0$~\cite{Pig25}.

In the following, unless differently specified, \emph{we always consider the \weak\ measure}.


\medskip

A \emph{one-counter automaton} (\oca, for short) is a pushdown automaton in which the pushdown alphabet
consists only of one symbol~$A$, besides~$Z_0$. The symbol~$Z_0$ is used only to mark the bottom. 
In this way, the pushdown store always contains a string of the form~$A^kZ_0$, representing the integer~$k\geq 0$, with
the exception of accepting configurations, in which the store is empty.
It is well known that the class of languages recognized by \ocas\ is a proper subclass of context-free languages and 
properly contains the class of regular languages.

\subsection{Semilinear sets, Presburger Arithmetic, Parikh Images}

In this subsection we present some of the tools that will be crucial to prove our results.

Let us denote by~$\IN$ the set of nonnegative integers and, for an integer~$m>0$, by~$\IN^m$ the set of vectors of~$m$ elements from~$\IN$,
with the usual operations of vector sum and multiplication of an integer by a vector.
A subset~$S\subseteq\IN^m$ is called \emph{linear} when is can be expressed as 
\[
	S = \{{\mathbf x}_0+h_1{\mathbf x}_1+\cdots+h_k{\mathbf x}_k \mid h_1,\ldots,h_k\in\IN \},
\]
for some vectors~${\mathbf x}_0, {\mathbf x}_1, \ldots, {\mathbf x}_k\in\IN^m$.
A \emph{semilinear set} in~$\IN^m$ is a finite union of linear subsets of~$\IN^m$.

\medskip

Semilinear sets can be characterized using \emph{Presburger Arithmetic}, which is defined
in terms of some formulas over integers, called \emph{Presburger formulas}~\cite{GS66b}.

A Presburger formula over~$\IN$ is defined from constants~$0$, $1$, and integer variables,
by using the operator~$+$, the predicate~$=$, the logical connective~$\vee, \wedge, \neg$, with
the standard meaning, together with the quantifiers~$\exists$ and~$\forall$ that can be applied
only to variables.

As proved in~\cite{GS66b}, the family of semilinear sets over~$\IN^m$ coincides with the family of sets defined by Presburger formula 
with~$m$ free variables. Furthermore, Presburger arithmetic is decidable, namely given a \emph{Presburger sentence}, i.e.,
a Presburger formula without free variables, it can be decided whether it is true~\cite{HB34}.

\begin{lemma}
\label{lemma:minSemilinear}
	Let set~$S\subseteq\IN^{1+m}$, $m\geq 0$, be a set.
	\begin{enumerate}
	\item If~$S$ is semilinear, then the following two sets are semilinear:
	\[
	S_{\min}=\{(i_0,i_1,\ldots,i_m)\in S\mid\forall k\in\IN,
	\mbox{ if $(k,i_1,\ldots,i_m)\in S$ then $i_0\leq k$}\}\,.
	\]
	\[
	S_{\max}=\{(i_0,i_1,\ldots,i_m)\in S\mid\forall k\in\IN,
	\mbox{ if $(k,i_1,\ldots,i_m)\in S$ then $i_0\geq k$}\}\,.
	\]
	\item Given~$\alpha_1,\alpha_2,\ldots,\alpha_m\in\IN$, if~$S$ is linear (semilinear, resp.), then the following set~$T\subseteq\IN^2$ is linear (semilinear, resp.):
	\[
	T = \{(i_0,\alpha_1i_1+\alpha_2i_2+\cdots+\alpha_mi_m) \mid (i_0,i_1,\ldots,i_m)\in S\}\,.
	\]
	\end{enumerate}
\end{lemma}
\begin{proof}
	We observe that the set~$S_{\min}$ is defined by taking, among all vectors in~$S$ that coincide on the last~$m$ components,
	the one having minimal first component.
	
	Let~$S$ be semilinear.
	Given a Presburger formula~$\phi(x_0,x_1,\ldots,x_m)$ corresponding to the~$S$,
	we consider the following formula
	\[
	\phi_{\min}(x_0,x_1,\ldots x_m)=\phi(x_0,x_1,\ldots x_m)\wedge(\forall y(\phi(y,x_1,\ldots x_m)\rightarrow(x_0\leq y)))\,.
	\]
	Since the predicate~$x\leq y$ can be expressed as~$\exists z(y=x+z)$ and~$f\rightarrow g$ as~$\neg f\vee g$, we can easily conclude that~$\phi_{\min}$ is a Presburger formula defining the set~$S_{\min}$.
	Hence, if~$S$ is semilinear then~$S_{\min}$ is semilinear.
	
	In a similar way, a vector~$(i_0,i_1,\ldots,i_m)$ belongs to the set ~$S_{\max}$ if~$i_0$ is the maximum value of the first component among
	all the vectors in~$S$ whose other components are~$i_1,\ldots,i_m$. In order to be defined, it is necessary that the number of such vectors is finite.
	The formula~$\phi_{\max}$ corresponding to~$S_{\max}$ can be easily obtained by replacing~$x_0\leq y$ by~$y\leq x_0$ in~$\phi_{\min}$.

\medskip

	We now prove the second statement in the linear case.	
	Suppose~$S =\{{\mathbf x}_0+h_1{\mathbf x}_1+\cdots+h_k{\mathbf x}_k \mid h_1,\ldots,h_k\in\IN \}$.
	For~$i=0,\ldots,k$, let~${\mathbf x}_i=(x_{i,0},x_{i,1},\ldots,x_{i,m})$ and consider the vector~${\mathbf y}_i\in\IN^2$
	defined as~${\mathbf y}_i=(x_{i,0}, \alpha_1x_{i,1}+\cdots+ \alpha_mx_{i,m})$.
	An easy verification shows that~~$T =\{{\mathbf y}_0+h_1{\mathbf y}_1+\cdots+h_k{\mathbf y}_k \mid h_1,\ldots,h_k\in\IN \}$.
	The extension to the semilinear case is obvious.
\qed
\end{proof}

The next lemma formalizes the fact that if an unbounded function~$f$ can be represented as a semilinear set, then it should linearly grow:

\begin{lemma}
\label{lemma:linearGrow}
	Let~$S\subseteq\IN^2$ be a semilinear set such that for each~$n\in\IN$ there exists at most one~$i_0\in\IN$ with~$(i_0,n)\in S$.
	Consider the function~$f_S:\IN\rightarrow\IN$ defined, for~$n\in\IN$, as
	\[
		f_S(n)=
		\left\{\begin{array}{ll}
			i_0	& \mbox{if } (i_0,n)\in S,\\
			0 & \mbox{otherwise}.
		\end{array}\right.
	\]	
	If~$f_S(n)\notin O(1)$ then~$f_S(n)\in O(n)\setminus o(n)$.\footnote{%
		We point out a subtle difference between~$f(n)\in\Theta(n)$ and~$f(n)\in O(n)\setminus o(n)$, for a function~$f:\IN\rightarrow\IN$.
		Besides~$f(n)\in O(n)$, in the first case we require that~$f(n)\geq c\cdot n+d$ for \emph{each sufficiently large} integer~$n$
		and some constanst~$c,d$; in the second case, it is enough that~$f(n)\geq c\cdot n+d$ \emph{for infinitely many} integers~$n$.
	}
\end{lemma}
\begin{proof}
	Let us suppose that~$f_S(n)$ is not bounded by any constants. First we prove that~$f_S(n)\notin o(n)$.
	Let us consider the case of a linear set~$S =\{{\mathbf x}_0+h_1{\mathbf x}_1+\cdots+h_k{\mathbf x}_k \mid h_1,\ldots,h_k\in\IN \}$, 
	where~${\mathbf x}_i=(x'_i,x''_i)$, for~$i=0,\ldots,k$.
	
	Since~$f_S(n)\notin O(1)$, there exists an~$i$, $1\leq i\leq k$, such that~$x'_i\neq 0$ and~$x''_i\neq 0$.
	By fixing~$h_j=0$, for~$j=1,\ldots,k$ with~$j\neq i$, for~$h_i\in\IN$ and~$n=x''_0+h_ix_i''$
	we obtain~$f_S(n)=x'_0+h_ix_i'$.
	By considering all values~$h_i\in\IN$, we obtain infinitely many integers~$n$ such that~$f_S(n)=x'_0+\frac{n-x''_0}{x''_i}\cdot x'_i$, 
	i.e., $f_S(n)$ it is linear in~$n$. This implies that~$f_S(n)\notin o(n)$.
	
	Now, let us suppose~$S=\bigcup_{i=1}^hS_i$ is semilinear, where~$S_1,S_2,\ldots,S_h$ are linear sets and~$h>0$.
	Since~$f_S(n)$ is not bounded by any constants, the cardinality of~$S$ cannot be finite. Hence, there exists a~$j$, $1\leq j\leq h$,
	such that the linear set~$S_j$, $1\leq j\leq h$, contains infinitely many pairs. 
	This allows to conclude that~$f_{S_j}(n)\notin o(n)$ and so~$f_S(n)\notin o(n)$.
	
	By refining this argument, we can also conclude that~$f_S(n)=O(n)$.
\qed
\end{proof}

We remind the reader that, given an alphabet~$\Sigma=\{a_1,a_2,\ldots,a_m\}$ of~$m$ symbols,
the \emph{Parikh map} $\psi: \Sigma^* \to \IN^m$ associates with each word $w \in \Sigma^*$ the vector
\[
\psi(w)=\left(|w|_{a_1}, |w|_{a_2}, \ldots, |w|_{a_m}\right)\,, 
\]
which counts the occurrences of each letter of~$\Sigma$ in~$w$. The vector~$\psi(w)$ is also called \emph{Parikh image} of $w$.
One can naturally generalize this map to each language~$L \subseteq \Sigma^*$, by defining
\[
\psi(L) = \{\psi(w) \mid w \in L\}\,.
\]
The set~$\psi(L)$ is called the \emph{Parikh image} of $L$.

A language is called \emph{semilinear} when its Parikh image is a semilinear set.

It is well known that the class of semilinear languages properly includes context-free languages and
that each semilinear set is the Parikh image of some regular language. 
This implies that for each semilinear language~$L$ there exists a regular language~$L'$ with the same Parikh image, i.e., $\psi(L)=\psi(L')$~\cite{Parikh1966}.

\medskip

In the case of a \emph{strictly letter-bounded language}~$L\subseteq a_1^*a_2^*\cdots a_m^*$, with $a_i\neq a_j$ for~$i\neq j$, there is a one-to-one correspondence between the strings in~$L$ and the vectors in~$\psi(L)$. Indeed, from the definition of Parikh image, it immediately follows
that for all integers~$i_1,i_2,\ldots,i_m\in\IN$ it holds that
\[
(i_1,i_2,\ldots,i_m) \in\psi(L)\mbox{ if and only if }a_1^{i_1}a_2^{i_2}\cdots a_m^{i_m}\in L\,.
\]
This also gives a one-to-one correspondence between Presburger formulas with~$m$ free variables and semilinear subsets of~$a_1^*a_2^*\cdots a_m^*$.
Given a formula~$\phi(x_1,x_2,\ldots,x_m)$, the corresponding language 
is
\[
L_\phi=\{a_1^{i_1}a_2^{i_2}\cdots a_m^{i_m}\mid\phi(i_1,i_2,\ldots,i_m)\mbox{ is true}\}\,.
\]
Notice that the language~$L_\phi$ is not necessarily a context-free language,
for instance~$L_{(x_1=x_2)\wedge(x_2=x_3)}=\{a_1^na_2^na_3^n\mid n\geq 0\}$.

\section{Decidability results for bounded languages}
\label{sec:decBounded}

The problem of deciding whether a \pda\ accepts in a number of turns bounded by some constant, under the \weak\ measure, is undecidable.
The problem remains undecidable when the constant is any fixed integer~$k>0$~\cite{Pig25}.

In this section we prove that in the restricted case of word-bounded languages these problems are decidable.
Furthermore, we show that the number of turns used to accept a string, when nonconstant, should linearly grow with respect to the length of the string.

In order to prove these results, we present a construction to modify any \pda\ in order to `count' in the input the number of turns.

\begin{lemma}
\label{lemma:countTurns}
	For each \pda~$\ma{M}$ with input alphabet~$\Sigma$, there exists a \pda~$\ma{M}_0$ with input alphabet~$\Sigma\cup\{a_0\}$, where~$a_0\notin\Sigma$,
	such that a string~$w\in(\Sigma\cup\{a_0\})^*$ is accepted by~$\ma{M}_0$ if and only if the string~$\pi_\Sigma(w)$ is accepted by a computation of~$\ma{M}$
	that makes exactly~$|w|_{a_0}$ turns.
	Furthermore, each computation of~$\ma{M}_0$ accepting an input~$w$ makes exactly~$|w|_{a_0}$ turns.
\end{lemma}
\begin{proof}
	First, we modify the given \pda~$\ma{M}$ in order to obtain an equivalent \pda~$\ma{M'}$ that keeps track in its finite control of the current `phase' 
	of the pushdown store, namely if it is increasing or decreasing.
	For this purpose, the set of states is the union of the set of states~$Q$ of~$\ma{M}$, with two marked copies of it denoted
	as~$\Qincr{Q}$ and~$\Qdecr{Q}$, used to remember if the machine is in an increasing or in a decreasing phase, respectively.
	The states in these copies will be indicated as~$\Qincr{q}$ and~$\Qdecr{q}$.
	The moves of~$\ma{M'}$ are defined from that of~$\ma{M}$ by updating the markers in the states as listed below,
	where~$p,q\in Q$, $A\in\Gamma$, $\sigma\in\Sigma\cup\{\varepsilon\}$:
	\begin{itemize}
		\item\emph{Increasing moves.}\\
		For~$(q,\alpha)\in\delta(p,\sigma,A)$, with~$\alpha\in\Gamma^*$ and~$|\alpha|>1$,
		the machine~$\ma{M'}$ has a transition setting the marker to~$\Qincr{}$, regardless the starting state is marked or not, 
		i.e., $\ma{M'}$ has the following transitions:
		\begin{itemize}
		\item $(\Qincr{q},\alpha)\in\delta'(p,\sigma,A)$,
		\item $(\Qincr{q},\alpha)\in\delta'(\Qincr{p},\sigma,A)$,
		\item $(\Qincr{q},\alpha)\in\delta'(\Qdecr{p},\sigma,A)$.
		\end{itemize}
		This allows to remember that the machine is in an increasing phase.
		\item\emph{Decreasing moves}.\\
		For~$(q,\varepsilon)\in\delta(p,\sigma,A)$, if the starting state is not marked, then the ending state is not marked, 
		otherwise it will be marked by~$\Qdecr{}$, regardless the marker in the starting state, i.e., $\ma{M'}$ has the following transitions:
		\begin{itemize}
		\item $(q,\varepsilon)\in\delta'(p,\sigma,A)$, 
		\item $(\Qdecr{q},\varepsilon)\in\delta'(\Qincr{p},\sigma,A)$, 
		\item $(\Qdecr{q},\varepsilon)\in\delta'(\Qdecr{p},\sigma,A)$.
		\end{itemize}
		In the first case the machine remembers that, up to now, no increasing (and so no decreasing) phase has been reached. 
		In the other cases, the machine remembers that the current phase is decreasing.
		\item\emph{Remaining moves}\\
		For~$(q,B)\in\delta(p,\sigma,A)$, with~$B\in\Gamma$, the marker, if any, is not changed, 
		i.e., $\ma{M'}$ has the following transitions:
		\begin{itemize}
		\item $(q,B)\in\delta'(p,\sigma,A)$,
		\item $(\Qincr{q},B)\in\delta'(\Qincr{p},\sigma,A)$,
		\item $(\Qdecr{q},B)\in\delta'(\Qdecr{p},\sigma,A)$.
		\end{itemize}
	\end{itemize}
	The machine~$\ma{M'}$ starts the computation in the initial state of~$\ma{M}$ and directly simulates the moves of~$\ma{M}$, 
	using states in~$Q$, until a move which increases the pushdown store is simulated.
	In this way, if the simulated computation of~$\ma{M}$ does not make any push and so any turn, then the corresponding computation 
	of~$\ma{M'}$ uses only original states from~$Q$.
	The first time a move increasing the store is simulated, 
	the machine starts to use the marked states, keeping track in the marker that the phase increasing. 
	This information is changed when a pop is executed, entering a decreasing phase, and so on.
	
	The \pda~$\ma{M}_0$ claimed in the statement of the lemma can be obtained by modifying~$\ma{M'}$ in such a way that~$\ma{M}_0$
	simulates all the moves of~$\ma{M'}$, including those that read input symbols, but it reads an extra
	input symbol~$a_0\notin\Sigma$ each time a transition from a state in~$\Qincr{Q}$ to a state in~$\Qdecr{Q}$ is executed.
	To this aim, each transition of the form~$(\Qdecr{q},\varepsilon)\in\delta'(\Qincr{p},\varepsilon,A)$, is replaced by the 
	transition~$(\Qdecr{q},\varepsilon)\in\delta_0(\Qincr{p},a_0,A)$, while each 
	transition~$(\Qdecr{q},\varepsilon)\in\delta'(\Qincr{p},\sigma,A)$ with~$\sigma\neq\varepsilon$ is replaced by the two
	transitions~$(\Qstay{q},A)\in\delta_0(\Qincr{p},\sigma,A)$ and~$(\Qdecr{q},\varepsilon)\in\delta_0(\Qstay{q},a_0,A)$, 
	where the state~$\Qstay{q}$ 
	is used to remember that~$\ma{M}_0$ has to read the symbol~$a_0$ and complete the simulation of a transition to~$\Qdecr{q}$ of~$\ma{M'}$.
	In this way, the strings in the language~$L_0\subseteq(\{a_0\}\cup\Sigma)^*$ accepted by~$\ma{M}_0$ have the property claimed in the statement of 
	the lemma.
	
	We point out that the computations of~$\ma{M}_0$ directly simulate those of the original \pda~$\ma{M}$, by making the same moves on the stack.
	So they use exactly the same number of turns as original computations. This number corresponds to the occurrences of the symbol~$a_0$ in~$w$.
\qed
\end{proof}

We are now ready to prove that it is possible to decide whether a \pda\ accepts a strictly letter-bounded language in a finite number of turns and whether it accepts a strictly letter-bounded language
in~$k$ turns, for any given integer~$k$.
This is in contrast with the general case, in which these properties are undecidable, even for \ocas, with the only exception of acceptance in~$0$ turns.

\begin{theorem}
\label{theorem:decLetterBounded}
	The following properties are decidable:
	\begin{itemize}
	\item whether a \pda\ accepts a strictly letter-bounded language in a finite number of turns,
	\item whether a \pda\ accepts a strictly letter-bounded language in~$k$ turns, for
	a given~$k\geq 0$.	
	\end{itemize}
\end{theorem}
\begin{proof}
 	Let~$\ma{M}$ be a \pda\ with input alphabet~$\Sigma$ of~$m$ symbols, accepting a language~$L\subseteq\Sigma^*$.
	
	First, we observe that, fixed a permutation~$a_1,a_2,\ldots,a_m$ of the symbols in~$\Sigma$, $L\subseteq a_1^*a_2^*\cdots a_m^*$ 
	if and only if~$L\cap(a_1^*a_2^*\cdots a_m^*)^c=\emptyset$.
	Since context-free languages are effectively closed under intersection with regular languages and emptiness  is
	decidable for context-free languages, then it is also decidable whether~$L\subseteq a_1^*a_2^*\cdots a_m^*$.
	To decide whether~$L$ is strictly-letter bounded, it is enough to test each permutation~$a_1,a_2,\ldots,a_m$,
	of the symbols in~$\Sigma$ up to find the one satisfying~$L\subseteq a_1^*a_2^*\cdots a_m^*$, if any.
	
	Once it is known that~$L\subseteq a_1^*a_2^*\cdots a_m^*$, let us consider the \pda~$\ma{M}_0$ obtained from~$\ma{M}$
	according to  Lemma~\ref{lemma:countTurns}, accepting a language~$L_0$. 
	Hence, a string~$a_1^{i_1}a_2^{i_2}\cdots a_m^{i_m}\in L$ is accepted by a computation of~$\ma{M}$ which makes~$k$ turns if and only 
	if~$(k,i_1,i_2,\ldots,i_m)\in\psi(L_0)$, where~$\psi(L_0)$ denotes the Parikh image of~$L_0$ which is semilinear.
	We observe that the smallest number of turns used to accept~$a_1^{i_1}a_2^{i_2}\cdots a_m^{i_m}$ is the smallest~$k$
	such that~$(k,i_1,i_2,\ldots,i_m)\in\psi(L_0)$.
	This leads to consider the following set:
	\begin{eqnarray}
	\psi(L_0)_{\min}=\{(i_0,i_1,\ldots,i_m)\in\psi(L_0) \mid \mbox{if $(k,i_1,\ldots,i_m)\in\psi(L_0)$ then $i_0\leq k$}\}\,
	\label{eq:psiMin}
	\end{eqnarray}
	which, according to Lemma~\ref{lemma:minSemilinear}, is semilinear and, hence,
	can be expressed by a formula~$\phi_{\min}(x_0,x_1,\ldots,x_n)$ in the Presburger Arithmetic.
	
	Let us now consider the following sentence:
	\[
	\phi(x) = \forall i_0\forall i_1\ldots\forall i_m(\phi_{\min}(i_0,i_1,\ldots,i_m) \rightarrow (i_0\leq x))\,.
	\]
	We observe that~$\ma{M}$ accepts in~$k$ turns, for any fixed~$k\geq 0$, if and only if the sentence~$\phi_k$ 
	obtained by replacing in~$\phi$ the variable~$x$ by the constant~$k$, i.e., $\phi_k=\phi(k)$, is true.	
	Furthermore,~$\ma{M}$ accepts in a finite number of turns if and only is the sentence~$\phi_{fin}=\exists k\,\phi(k)$ is true.
	Since Presburger Arithmetic is decidable, $\phi_k$ and~$\phi_{fin}$ are decidable.
	Hence, we conclude that the two properties are decidable.
\qed
\end{proof}

\noindent
We now extend Theorem~\ref{theorem:decLetterBounded} to the case of \pdas\ accepting word-bounded languages,
by reducing the word-bounded case to the strictly letter-bounded case, solved in Theorem~\ref{theorem:decLetterBounded}.
This is done by making use of an inverse homomorphism.
To this aim, let us prove the following result:

\begin{lemma}
	\label{lemma:inverseHomomorphism}
	Given an alphabet~$\Sigma$, let~$w_1,w_2,\ldots,w_m\in\Sigma^*$ be fixed strings.
	Consider an~$m$-symbol alphabet~$\Delta=\{a_1,a_2,\ldots,a_m\}$ and the 
	morphism~$h:\Delta\rightarrow\Sigma^*$ defined as~$h(a_i)=w_i$, $i=1,\ldots,m$.
	Given~$L\subseteq w_1^*w_2^*\cdots w_m^*$, let~$L_h=h^{-1}(L)\cap(a_1^*a_2^*\cdots a_m^*)$.
	Then:
	\begin{itemize}
	\item $L_h=\{a_1^{i_1}a_2^{i_2}\cdots a_m^{i_m} \mid w_1^{i_1}w_2^{i_2}\cdots w_m^{i_m}\in L\}$\,.
	\item Given a \pda~$\ma{M}$ accepting the language~$L$, we can construct a \pda~$\ma{M}_h$ such that there
	is a one-to-one correspondence between the computations of~$\ma{M}_h$ on each string~$x\in a_1^*a_2^*\cdots a_m^*$ and
	the computations of~$\ma{M}$ on the string~$h(x)$, in such a way that corresponding computations
	use the same number of turns, i.e., for integer~$k\geq 0$,~$\ma{M}_h$ has a computation on~$x$ that makes~$k$ turns
	if and only if~$\ma{M}$ has a computation on~$h(x)$ that makes~$k$ turns.
	\item For each~$x\in L_h$ and for each~$k\geq 0$,
	$\ma{M}_h$ acceptes~$x$ in~$k$ turns if and only if~$\ma{M}$ accepts~$h(x)$ is~$k$ turns.
	\end{itemize}
\end{lemma}
\begin{proof}
	We observe that from the definition of~$h$ it follows 
	that~$h(a_1^{i_1}a_2^{i_2}\cdots a_m^{i_m})=w_1^{i_1}w_2^{i_2}\cdots w_m^{i_m}$,
	for all~$i_1,i_2,\ldots,i_m\in\IN$. From this, we easily obtain the first statement.
	
	The construction of the \pda~$\ma{M}_h$, we now outline, is obtained by adapting the construction used in~\cite[Thm.~6.3]{HU79}
	to prove that the class of context-free languages is closed under inverse homomorphism.
	
	The \pda~$\ma{M}_h$ is obtained by replacing the input tape of~$\ma{M}$ by a finite state control~$\ma{A}$
	that reads a string in~$\Delta^*$ from the `new' input tape and by a finite buffer.
	The state control of~$\ma{A}$ and the buffer are embedded in the finite control of~$\ma{M}_h$
	together with the control of~$\ma{M}$.
	
	The component~$\ma{A}$ has to compute~$h(x)$, where~$x\in\Delta^*$ is the input string of~$\ma{M}_h$, 
	while verifying that~$x\in a_1^*a_2^*\cdots a_m^*$ (if this is not the case, then~$\ma{M}_h$ rejects).
	The string~$h(x)$ is used as input for~$\ma{M}$. The part of~$\ma{M}_h$ corresponding to~$\ma{M}$,
	directly simulates a computation of~$\ma{M}$ and makes exactly the same turns as the original
	computation on~$h(x)$ (this gives the one-to-one correspondence with original computations).
	The string~$h(x)$ is not saved by the machine, but the parts~$\ma{A}$ and~$\ma{M}$
	work together using a `producer/consumer' scheme: Each time~$\ma{A}$ reads a symbol~$a_i$ from the input,
	$1\leq i\leq m$, it saves the string $w_i$ in the buffer. The part corresponding to~$\ma{M}$ 
	reads symbols from the buffer, instead of from the input tape. When the buffer is empty, 
	$\ma{A}$ reads another input symbol and saves the corresponding string in the buffer to be used by~$\ma{M}$,
	unless there are no more input symbols for~$\ma{M}_h$. In this case, the computation stops
	accepting or rejecting, according to the simulated computation of~$\ma{M}$.
	
	Since along the simulation the pushdown contents of~$\ma{M}_h$ and of~$\ma{M}$ are the same, we can
	easily conclude that for each~$x\in L_h$ and for each~$k\geq 0$,
	$\ma{M}_h$ accepts~$x$ in~$k$ turns if and only if~$\ma{M}$ accepts~$h(x)$ in~$k$ turns.
\qed
\end{proof}

We point out that in the statement of Lemma~\ref{lemma:inverseHomomorphism} it is required that the
cardinality of~$\Delta$ is~$m$. This implies that~$a_i\neq a_j$, for~$i\neq j$, even in case~$w_i=w_j$,
namely the language~$L_h$ is
\emph{stricty} letter bounded, even if the words~$w_1,w_2,\ldots,w_m\in\Sigma^*$
are not pairwise different.

We also observe that it could happen to have different strings~$x,y\in L_h$ 
such that~$h(x)=h(y)=w$, for some~$w\in L$. This would imply that each accepting computation on~$w$
corresponds to an accepting computation on~$x$ and to an accepting computation on~$y$, 
even though~$x$ and~$y$ may even be of different lengths. 
Hence, the number of turns that are sufficient to~$\ma{M}_h$ to
accept the strings~$x$ and~$y$ are the same as that of~$\ma{M}$ to accept~$w$.

\smallskip

We are now able to extend Theorem~\ref{theorem:decLetterBounded} to the word-bounded case:

\begin{theorem}
\label{theorem:decWordBounded}
	The following properties are decidable:
	\begin{itemize}
	\item whether a \pda\ accepts a word-bounded language in a finite number of turns,
	\item whether a \pda\ accepts a word-bounded language in~$k$ turns, for
	a given~$k\geq 0$.	
	\end{itemize}
\end{theorem}
\begin{proof}
	Let~$\ma{M}$ be a \pda\ with input alphabet~$\Sigma$. 
	First of all, it can be decided if the language~$L$ accepted by~$\ma{M}$ is
	word bounded. Furthermore, in this case, strings~$w_1,w_2,\ldots,w_m\in\Sigma^*$, such that~$L\subseteq w_1^*w_2^*\cdots w_m^*$ can
	be effectively found~\cite[Thm.~5.2]{GS64}.
	
	According to Lemma~\ref{lemma:inverseHomomorphism}, from~$\ma{M}$ we can construct a \pda~$\ma{M}_h$
	accepting a strictly letter-bounded language~$L_h=h^{-1}(L)\cap(a_1^*a_2^*\cdots a_m^*)$, where~$\Delta=\{a_1,\ldots,a_m\}$ is
	a~$m$-symbol alphabet and~$h(a_i)=w_i$, $i=1,\ldots,m$, such that for each~$x\in L_h$ and~$k\geq 0$,
	$\ma{M}_h$ accepts~$x$ in~$k$ turns if and only if~$\ma{M}$ accepts~$h(x)$ is~$k$ turns.
	Hence, the result follows from Theorem~\ref{theorem:decLetterBounded}.
\qed
\end{proof}

We point out that the decidability questions considered in Theorem~\ref{theorem:decWordBounded} concern the machines, 
not the accepted languages. Actually, it is well-known that each word-bounded language can be accepted by a \pda\ making 
a number of turns bounded by a constant~\cite{GS66}.
In particular, if the accepted language is a subset of~$w_1^*w_2^*\cdots w_m^*$, then~$m-1$ turns are sufficient.
However, we could have \pdas\ accepting the same language using an higher, even unbounded, number of turns.
For results on \pdas\ accepting bounded languages and the reduction of the number
of turns, we address the reader to~\cite{MP13}. 

\medskip

We are now going to prove that if the number of turns used by a \pda\ to accept a word-bounded language 
is not limited by any constants, then it should grow at least as a linear function:
	
\begin{theorem}
\label{th:lowerBound}
	Let~$\ma{M}$ a \pda\ accepting a word-bounded language~$L$ in~$t(n)$ turns (\weak\ measure).
	If~$t(n)\notin O(1)$ then~$t(n)\in O(n)\setminus o(n)$.
\end{theorem}
\begin{proof}
	According to Definition~\ref{def:turns}, let us consider the function
	\[
		t(n)=
		\left\{\begin{array}{ll}
			\max\{t(x)\mid|x|=n\mbox{ and }x\in L\}&\mbox{if $L$ contains at least one string of length~$n$},\\
			0 & \mbox{otherwise},
		\end{array}\right.
	\]	
	where, for~$x\in L$, $t(x)$ is the minimum number of turns made by~$\ma{M}$ in the computations that accept~$x$.	
	We point out that, since for each~$n\geq 0$ the number of strings of length~$n$ is finite, the value~$t(n)$ is always
	defined.
	To prove the result, we are going to show that the following set is semilinear:
	\[
		\{(t(n),n)\mid~L \mbox{ contains at least one string of length~$n$}\}\,.
	\]
	Thus, the result will follow from Lemma~\ref{lemma:linearGrow}.

	Let us start by supposing that the
	\pda~$\ma{M}$ accepts a \emph{strictly letter-bounded} language~$L\subseteq a_1^*a_2^*\cdots a_m^*$.
	
	By exploiting the same idea of the proof of Theoreom~\ref{theorem:decLetterBounded},
	from~$\ma{M}$ we build a \pda~$\ma{M}_0$ that, using an extra input symbol~$a_0$, 
	counts the number of turns made in accepting computations. Let~$L_0$ be the
	language accepted by~$\ma{M}_0$. Then, for any string~$a_1^{i_1}a_2^{i_2}\cdots a_m^{i_m}\in L$, 
	we have that~$t(a_1^{i_1}a_2^{i_2}\cdots a_m^{i_m})=\min\{i_0\mid(i_0,i_1,\ldots,i_k)\in\psi(L_0)\}$.
	If~$L$ contains some string of length~$n$, then:
	\begin{eqnarray*}
	t(n)&=&\max\{t(x)\mid|x|=n\mbox{ and }x\in L\} \\
		&=& \max\{i_0\mid (i_0,i_1,\ldots,i_m)\in\psi(L_0)_{\min}\mbox{ and }n=i_1+\cdots+i_m\}\,,
	\end{eqnarray*}
	where~$\psi(L_0)_{\min}$ is defined in~(\ref{eq:psiMin}).
	
	In the word-bounded case, i.e.,~$L\subseteq w_1^*w_2^*\cdots w_m^*$, where~$w_1,w_2,\ldots,w_k$ are fixed words on an alphabet~$\Sigma$, we first use Lemma~\ref{lemma:inverseHomomorphism} as follows.
	We introduce the~$m$-symbol alphabet~$\Delta=\{a_1,a_2,\ldots,a_m\}$ and the 
	morphism~$h:\Delta\rightarrow\Sigma^*$ defined by~$h(a_i)=w_i$, $i=1,\ldots,m$.
	From the  \pda~$\ma{M}$ we build a \pda~$\ma{M}_h$ recognizing the 
	language~$L_h=h^{-1}(L)\cap(a_1^*a_2^*\cdots a_m^*)$ and such that for each~$x\in L_h$ and~$k\geq 0$, 
	$\ma{M}_h$ has an accepting computation on input~$x$ which makes~$k$ turns if and only if~$\ma{M}$ has an accepting computation on
	input~$h(x)$ that makes~$k$ turns.
	We then apply the above-described process for the strictly-letter bounded case, starting
	from the \pda~$\ma{M}_h$, hence obtaining a \pda~$\ma{M}_0$.
	Let us denote~$t_h$ the function giving the number of turns made by~$\ma{M}_h$.
	For~$a_1^{i_1}\cdots a_m^{i_m}\in L_h$ we have 
	that~$t_h(a_1^{i_1}\cdots a_m^{i_m})=t(w_1^{i_1}\cdots w_m^{i_m})$.
	Furthermore~$|w_1^{i_1}\cdots w_m^{i_m}|=|w_1|\cdot i_1+\cdots+|w_m|\cdot i_m$.
	This allows to conclude that if~$L$ contains some string of length~$n$ then:
	\begin{eqnarray*}
	t(n)&=&\max\{t(x)\mid|x|=n\mbox{ and }x\in L\} \\
			&=& \max\{i_0 \mid (i_0,n)\in T\}\,,
	\end{eqnarray*}
	where
	\begin{eqnarray*}
		T=\{(i_0,n) \mid (i_0,i_1,\ldots,i_m)\in\psi(L_0)_{\min}\mbox{ and }n=|w_1|\cdot i_1+\cdots+|w_m|\cdot i_m\}\,.
	\end{eqnarray*}
	Begin~$\psi(L_0)_{\min}$ semilinear, according to Lemma~\ref{lemma:minSemilinear}, also the set~$T$ and the 
	set~$T_{\max}=\{(i_0,n)\in T\mid\mbox{ if $(k,n)\in T$ then~$i_0\geq k$}\}$
	are semilinear.
	
	Furthermore, $t(n) = i_0$ if and only if~$(i_0,n)\in T_{\max}$.
	According to Lemma~\ref{lemma:linearGrow} we then
	conclude that if~$t(n)\notin O(1)$ then~$t(n)\notin o(n)$.
\qed
\end{proof}

\section{Beyond the \weak\ case for bounded languages}
\label{sec:accept}

In the previous section we have shown that in the case of \pdas\ accepting bounded languages,
if the number of turns under the \weak\ measure is not bounded by any constants, then it
linearly grows with respect to the length of the input. This result does not hold in the
case of general languages, namely by removing the restriction to bounded languages.
Indeed, in~\cite{Pig25} it was proved that for each integer~$k$ there exists a language
for which a number of turns of the order of~$\log^{(k)}$ (the composition of the logarithm 
function~$k$ times with itself) is necessary and sufficient, so obtaining an infinite hierarchy, 
with a nonconstant sublinear number of turns.
However, as we prove in this section, if we consider \emph{all accepting} computations, then
when the number of turns is nonconstant, it linearly grows with respect to the length of the input.

\begin{theorem}
\label{th:accept}
	Let~$\ma{M}$ a \pda\ accepting a language~$L\subseteq\Sigma^*$.
	Let~$t(n)$ be the maximum number of turns made by accepting computations on inputs of length~$n$ 
	(\accept\ measure), if any\footnote{In case it does not exist any accepted input of length~$n$, we can stipulate~$t(n)=0$-}.
	If~$t(n)\notin O(1)$ then~$t(n)\notin o(n)$.
	Furthermore, if the number of accepting computations on each~$w\in L$ is
	finite, then~$t(n)=O(n)$.
\end{theorem}
\begin{proof}
	According to Lemma~\ref{lemma:countTurns}, from~$\ma{M}$ we can obtain a \pda~$\ma{M}_0$, 
	with input alphabet~$\Sigma\cup\{a_0\}$, where~$a_0\notin\Sigma$, such that a 
	string~$w\in(\Sigma\cup\{a_0\})^*$ is accepted by~$\ma{M}_0$ if and only if the 
	string~$\pi_\Sigma(w)$ is accepted by a computation of~$\ma{M}$ that makes 
	exactly~$|w|_{a_0}$ turns.
	We modify the input alphabet of~$\ma{M}_0$ to be~$\{a_0,a\}$: in each transition that
	reads a symbol~$\sigma\in\Sigma$, the input symbol is replaced by~$a$.
	In this way, a string~$w$ is accepted by the resulting machine if and only if the original
	\pda~$\ma{M}$ has an accepting computation on an input of length~$|w|_a$ that makes~$|w|_{a_0}$
	turns, namely the Parikh image of the language accepted by~$\ma{M}_0$, after this modification,
	is the set
	\[
	S = \{(t,n) \mid \ma{M} \mbox{ has an accepting computation on an input of length~$n$
	that makes~$t$ turns}\}\,,
	\]
	which, so, is semilinear. 
	Hence, fixed~$n$, the maximum number of turns to accept inputs of length~$n$ is the maximum~$t$ such that~$(t,n)\in S$.
	This leads to consider the following set which,
	according to Lemma~\ref{lemma:minSemilinear}, is semilinear:
	\[
	S_{\max}= \{(i_0,n)\in S \mid \forall t\in\IN, \mbox{ if $(t,n)\in S$ then $i_0\geq k$}\}\,.
	\]
	In other words,~$(i_0,n)\in S_{\max}$ implies~$t(n)=i_0$.
	Thus, from Lemma~\ref{lemma:linearGrow} it follows that if~$t(n)\notin O(1)$ then~$t(n)\notin o(n)$.
	We also observe that if for some~$n\in\IN$ there is some input of length~$n$ with infinitely many accepting
	computations, the value~$t(n)$ could be infinite. In this case the set~$S_{\max}$ does not
	contain any pairs with second component~$n$ (i.e., on such~$n$ the value of the function~$f_{S_{\max}}$
	defined in Lemma~\ref{lemma:linearGrow} is~$0$ and does not represent the value~$t(n)$.
	Clearly, in that case we cannot say that~$t(n)$ is linearly bounded by~$n$.
	However, if for each accepted input the number of accepting computations is finite,
	then Lemma~\ref{lemma:linearGrow} allows to conclude that~$t(n)=O(n)$.
\qed
\end{proof}

If for some input there are infinitely many accepting computations that arbitrarily increase
the number of turns, then each of these computations should contain some ``useless'' sequence of 
$\varepsilon$-moves that could be skipped to obtain a shorter ``reduced'' computation, i.e., 
a computation that on ``real'' input symbols makes the same sequence of transitions. If we
consider only reduced accepting computations, we can upper limit the number of turns
by a linear function, as in Theorem~\ref{th:accept}. We address the reader to~\cite[Sect.~8]{PP25}
for a similar discussion.

Since in the deterministic case each input can have at most one accepting computation and so
the \accept\ and the \weak\ notion coincide, from Theorem~\ref{th:accept} we obtain:

\begin{corollary}
	Let~$\ma{M}$ be a \emph{deterministic} pushdown automaton accepting each input of
	length~$n$ in~$t(n)$ turns (\weak\ measure). If~$t(n)\notin O(1)$ then~$t(n)\in O(n)\setminus o(n)$.
\end{corollary}

We can extend the results of this section to the \strong\ case, i.e., by considering
\emph{all computations}, regardless the fact that they are accepting or rejecting.
The rough idea is to modify the given \pda\ in order to make each computation accepting,
at the possible price of one extra turn, and use the results for the \accept\ case.
However, this requires to take into considerations computations that stop
before reading the entire input and computations entering infinite loops.
This can be done applying the techniques explained in~\cite[Sect.~7]{PP25}.

\section{Final considerations}

We proved that the number of turns made by a \pda\ to recognize a bounded context-free language, when not limited
by any constants, should linearly grow with respect to the length of the input, even if we consider, for each accepted input, only the less expensive accepting
computation (\weak\ measure). As proved in~\cite{Pig25}, the same does not hold if we remove the restriction to
bounded languages. However, if for each accepted input we consider the most expensive accepting computation
(\accept\ measure), then we obtain again a linear grow even in the general case.

One could ask if the same results hold for more general devices.
We point out that \emph{two-counter automata} (namely finite automata with a one-way input tape and two counters,
that can be increased, decreased and tested by zero), can recognize the set of powers of
two written in unary notation, i.e., the language
\[
L_{p} = \{a^{2^k}\mid k\geq 0\}
\]
by reversing the counters from an increment to a decrement phase a logarithmic number of times,
i.e., in~$O(\log n)$ turns, using a \emph{deterministic} machine.
So the linear lower bound for the \weak\ measure in the case of bounded languages, as well as the
linear lower bound for the \accept\ case for general languages, do not hold if the pushdown store is replaced by 
two counters.

A related result has been proved in~\cite{FKMMZ23} for machines with a fixed number of counters.
In that paper, the authors considered turns (called \emph{reversals}) in counter machines (without input).
They prove that in a computation consisting of~$n$ steps from the initial configuration to some fixed target configuration,
the number of turns, when not bounded by any constants, should grow at least as~$\log n$.
Furthermore such a bound can be reached. \emph{All computations} of~$n$ steps leading to the target
configuration are taken into account, hence the \accept\ measure is used.

Another possible generalization is to machines that can scan the input tape in both directions.
It is well-known that this possibility does not increase the computational power of finite automata, even if
it can be useful to have smaller representations
(for recent results on the descriptional complexity of two-way finite automata accepting letter-bounded language 
we address the reader to~\cite{CPP25}).
However, in the case of more powerful devices, the possibility of scanning the input tape in
both directions can strongly increase the computational power.
We can observe that the above-mentioned language~$L_p$ is accepted by a deterministic one-counter automaton with
a two-way input tape, using a logarithmic number of turns and a logarithmic number of reversals of
the input head.

\bibliographystyle{eptcs}
\bibliography{biblio}

\begin{thebibliography}{10}
\providecommand{\bibitemdeclare}[2]{}
\providecommand{\surnamestart}{}
\providecommand{\surnameend}{}
\providecommand{\urlprefix}{Available at }
\providecommand{\url}[1]{\texttt{#1}}
\providecommand{\href}[2]{\texttt{#2}}
\providecommand{\urlalt}[2]{\href{#1}{#2}}
\providecommand{\doi}[1]{doi:\urlalt{https://doi.org/#1}{#1}}
\providecommand{\eprint}[1]{arXiv:\urlalt{https://arxiv.org/abs/#1}{#1}}
\providecommand{\bibinfo}[2]{#2}

\bibitemdeclare{inproceedings}{CPP25}
\bibitem{CPP25}
\bibinfo{author}{Alessandro \surnamestart {Clerici Lorenzini}\surnameend},
  \bibinfo{author}{Giovanni \surnamestart Pighizzini\surnameend} \&
  \bibinfo{author}{Luca \surnamestart Prigioniero\surnameend}
  (\bibinfo{year}{2025}): \emph{\bibinfo{title}{Two-Way Automata and Bounded
  Languages}}.
\newblock In \bibinfo{editor}{Giuseppa \surnamestart Castiglione\surnameend} \&
  \bibinfo{editor}{Sabrina \surnamestart Mantaci\surnameend}, editors:
  {\slshape \bibinfo{booktitle}{Implementation and Application of Automata -
  29th International Conference, {CIAA} 2025, Palermo, Italy, September 22-25,
  2025, Proceedings}}, {\slshape \bibinfo{series}{Lecture Notes in Computer
  Science}} \bibinfo{volume}{15981}, \bibinfo{publisher}{Springer}, pp.
  \bibinfo{pages}{73--85}, \doi{10.1007/978-3-032-02602-6\_6}.

\bibitemdeclare{inproceedings}{FKMMZ23}
\bibitem{FKMMZ23}
\bibinfo{author}{Alain \surnamestart Finkel\surnameend},
  \bibinfo{author}{Shankara~Narayanan \surnamestart Krishna\surnameend},
  \bibinfo{author}{Khushraj \surnamestart Madnani\surnameend},
  \bibinfo{author}{Rupak \surnamestart Majumdar\surnameend} \&
  \bibinfo{author}{Georg \surnamestart Zetzsche\surnameend}
  (\bibinfo{year}{2023}): \emph{\bibinfo{title}{Counter Machines with
  Infrequent Reversals}}.
\newblock In \bibinfo{editor}{Patricia \surnamestart Bouyer\surnameend} \&
  \bibinfo{editor}{Srikanth \surnamestart Srinivasan\surnameend}, editors:
  {\slshape \bibinfo{booktitle}{43rd {IARCS} Annual Conference on Foundations
  of Software Technology and Theoretical Computer Science, {FSTTCS} 2023,
  {IIIT} Hyderabad, Telangana, India, December 18-20, 2023}}, {\slshape
  \bibinfo{series}{LIPIcs}} \bibinfo{volume}{284}, \bibinfo{publisher}{Schloss
  Dagstuhl - Leibniz-Zentrum f{\"{u}}r Informatik}, pp.
  \bibinfo{pages}{42:1--42:17}, \doi{10.4230/LIPICS.FSTTCS.2023.42}.

\bibitemdeclare{article}{GS64}
\bibitem{GS64}
\bibinfo{author}{Seymour \surnamestart Ginsburg\surnameend} \&
  \bibinfo{author}{Edwin~H. \surnamestart Spanier\surnameend}
  (\bibinfo{year}{1964}): \emph{\bibinfo{title}{Bounded Algol-like Languages}}.
\newblock {\slshape \bibinfo{journal}{Transactions of the American Mathematical
  Society}} \bibinfo{volume}{113}(\bibinfo{number}{2}), p. \bibinfo{pages}{333
  – 368}, \doi{10.1090/S0002-9947-1964-0181500-1}.

\bibitemdeclare{article}{GS66}
\bibitem{GS66}
\bibinfo{author}{Seymour \surnamestart Ginsburg\surnameend} \&
  \bibinfo{author}{Edwin~H. \surnamestart Spanier\surnameend}
  (\bibinfo{year}{1966}): \emph{\bibinfo{title}{Finite-Turn Pushdown
  Automata}}.
\newblock {\slshape \bibinfo{journal}{SIAM Journal on Control}}
  \bibinfo{volume}{4}(\bibinfo{number}{3}), pp. \bibinfo{pages}{429--453},
  \doi{10.1137/0304034}.

\bibitemdeclare{article}{GS66b}
\bibitem{GS66b}
\bibinfo{author}{Seymour \surnamestart Ginsburg\surnameend} \&
  \bibinfo{author}{Edwin~H. \surnamestart Spanier\surnameend}
  (\bibinfo{year}{1966}): \emph{\bibinfo{title}{Semigroups, {P}resburger
  formulas, and languages}}.
\newblock {\slshape \bibinfo{journal}{Pacific Journal of Mathematics}}
  \bibinfo{volume}{16}(\bibinfo{number}{2}), pp. \bibinfo{pages}{285--296},
  \doi{10.2140/pjm.1966.16.285}.

\bibitemdeclare{book}{HB34}
\bibitem{HB34}
\bibinfo{author}{David \surnamestart Hilbert\surnameend} \&
  \bibinfo{author}{Paul \surnamestart Bernays\surnameend}
  (\bibinfo{year}{1934}): \emph{\bibinfo{title}{{Grundlagen der Mathematik ---
  Erster Band}}}.
\newblock {\slshape \bibinfo{series}{{Grundlehren der mathematischen
  Wissenschaften}}}~\bibinfo{volume}{XL}, \bibinfo{publisher}{Springer,
  Berlin}.

\bibitemdeclare{book}{HU79}
\bibitem{HU79}
\bibinfo{author}{J.~E. \surnamestart Hopcroft\surnameend} \&
  \bibinfo{author}{J.~D. \surnamestart Ullman\surnameend}
  (\bibinfo{year}{1979}): \emph{\bibinfo{title}{Introduction to Automata
  Theory, Languages and Computation}}.
\newblock \bibinfo{publisher}{Addison-Wesley}.

\bibitemdeclare{article}{Mal07}
\bibitem{Mal07}
\bibinfo{author}{Andreas \surnamestart Malcher\surnameend}
  (\bibinfo{year}{2007}): \emph{\bibinfo{title}{On Recursive and Non-recursive
  Trade-Offs between Finite-Turn Pushdown Automata}}.
\newblock {\slshape \bibinfo{journal}{J. Autom. Lang. Comb.}}
  \bibinfo{volume}{12}(\bibinfo{number}{1-2}), pp. \bibinfo{pages}{265--277},
  \doi{10.25596/JALC-2007-265}.

\bibitemdeclare{article}{MP13}
\bibitem{MP13}
\bibinfo{author}{Andreas \surnamestart Malcher\surnameend} \&
  \bibinfo{author}{Giovanni \surnamestart Pighizzini\surnameend}
  (\bibinfo{year}{2013}): \emph{\bibinfo{title}{Descriptional complexity of
  bounded context-free languages}}.
\newblock {\slshape \bibinfo{journal}{Inf. Comput.}} \bibinfo{volume}{227}, pp.
  \bibinfo{pages}{1--20}, \doi{10.1016/J.IC.2013.03.008}.

\bibitemdeclare{article}{Parikh1966}
\bibitem{Parikh1966}
\bibinfo{author}{Rohit \surnamestart Parikh\surnameend} (\bibinfo{year}{1966}):
  \emph{\bibinfo{title}{On Context-Free Languages}}.
\newblock {\slshape \bibinfo{journal}{J. {ACM}}}
  \bibinfo{volume}{13}(\bibinfo{number}{4}), pp. \bibinfo{pages}{570--581},
  \doi{10.1145/321356.321364}.

\bibitemdeclare{inproceedings}{Pig25}
\bibitem{Pig25}
\bibinfo{author}{Giovanni \surnamestart Pighizzini\surnameend}
  (\bibinfo{year}{2025}): \emph{\bibinfo{title}{Turn Complexity of Context-Free
  Languages, Pushdown and One-Counter Automata}}.
\newblock In \bibinfo{editor}{Sang{-}Ki \surnamestart Ko\surnameend} \&
  \bibinfo{editor}{Florin \surnamestart Manea\surnameend}, editors: {\slshape
  \bibinfo{booktitle}{Developments in Language Theory - 29th International
  Conference, {DLT} 2025, Seoul, South Korea, August 19-22, 2025,
  Proceedings}}, {\slshape \bibinfo{series}{Lecture Notes in Computer Science}}
  \bibinfo{volume}{16036}, \bibinfo{publisher}{Springer}, pp.
  \bibinfo{pages}{77--91}, \doi{10.1007/978-3-032-01475-7\_6}.
\newblock \bibinfo{note}{An extended version is available on ArXiv at
  \url{https://arxiv.org/abs/2603.08331}}.

\bibitemdeclare{article}{Pig26}
\bibitem{Pig26}
\bibinfo{author}{Giovanni \surnamestart Pighizzini\surnameend}
  (\bibinfo{year}{2026}): \emph{\bibinfo{title}{Push complexity: Optimal bounds
  and decidability}}.
\newblock {\slshape \bibinfo{journal}{Theor. Comput. Sci.}}
  \bibinfo{volume}{1071}, p. \bibinfo{pages}{115839},
  \doi{10.1016/J.TCS.2026.115839}.

\bibitemdeclare{article}{PP23}
\bibitem{PP23}
\bibinfo{author}{Giovanni \surnamestart Pighizzini\surnameend} \&
  \bibinfo{author}{Luca \surnamestart Prigioniero\surnameend}
  (\bibinfo{year}{2023}): \emph{\bibinfo{title}{Pushdown automata and constant
  height: decidability and bounds}}.
\newblock {\slshape \bibinfo{journal}{Acta Informatica}}
  \bibinfo{volume}{60}(\bibinfo{number}{2}), pp. \bibinfo{pages}{123--144},
  \doi{10.1007/S00236-022-00434-0}.

\bibitemdeclare{article}{PP25}
\bibitem{PP25}
\bibinfo{author}{Giovanni \surnamestart Pighizzini\surnameend} \&
  \bibinfo{author}{Luca \surnamestart Prigioniero\surnameend}
  (\bibinfo{year}{2025}): \emph{\bibinfo{title}{Pushdown and one-counter
  automata: Constant and non-constant memory usage}}.
\newblock {\slshape \bibinfo{journal}{Inf. Comput.}} \bibinfo{volume}{306}, p.
  \bibinfo{pages}{105329}, \doi{10.1016/J.IC.2025.105329}.

\end{thebibliography}

\end{document}